\documentclass[11pt,a4paper]{article}

\usepackage{fullpage}
\usepackage{times}

\usepackage{color}
\usepackage{amssymb,amsmath,stmaryrd}
\usepackage[lined,boxed,commentsnumbered]{algorithm2e}

\newtheorem{theorem}{Theorem}

\newtheorem{lemma}[theorem]{Lemma}

\newtheorem{corollary}[theorem]{Corollary}

\newcommand{\sq}{\hbox{\rlap{$\sqcap$}$\sqcup$}}
\newcommand{\qed}{\hspace*{\fill}\sq}
\newenvironment{proof}{\noindent {\bf Proof.}\ }{\qed\par\vskip 4mm\par}

\def\Acoord{\texttt{ACOORD}}
\def\Bcoord{\texttt{BCOORD}}
\def\Ccoord{\texttt{CCOORD}}
\def\Dcoord{\texttt{DCOORD}}
\def\balance{\texttt{Balance}}

\newcommand{\opt}[1]{ {\widetilde{#1}} }

\newcommand{\norm}[1]{{\Vert#1\Vert}}

\allowdisplaybreaks

\begin{document}

\title{\bf An almost ideal coordination mechanism\\for unrelated machine scheduling\thanks{
This work was partially supported by the Caratheodory grant E.114 from the University of Patras and the project ANR-14-CE24-0007-01 \emph{``CoCoRICo-CoDec"}. Part of the work was done while the second author was visiting the Institute for Mathematical Sciences, National University of Singapore in 2015.}
}

\author{Ioannis Caragiannis\thanks{Computer Technology Institute ``Diophantus'' \& Department of Computer Engineering and Informatics, University of Patras, 26504 Rion, Greece. Email: {\tt caragian@ceid.upatras.gr}} \and Angelo Fanelli\thanks{CNRS (UMR-6211), Universit\'e de Caen Basse-Normandie, France. Email: {\tt angelo.fanelli@unicaen.fr}}}


\date{}

\maketitle

\begin{abstract}
Coordination mechanisms aim to mitigate the impact of selfishness when scheduling jobs to different machines. Such a mechanism defines a scheduling policy within each machine and naturally induces a game among the selfish job owners. The desirable properties of a coordination mechanism includes simplicity in its definition and efficiency of the outcomes of the induced game. We present a broad class of coordination mechanisms for unrelated machine scheduling that are simple to define and we identify one of its members (mechanism \Dcoord) that is superior to all known mechanisms. \Dcoord\ induces potential games with logarithmic price of anarchy and only constant price of stability. Both bounds are almost optimal.
\end{abstract}

\section{Introduction}\label{sec:intro}
We consider a selfish scheduling setting where each job owner acts as a non-cooperative player and aims to assign her job to one of the available machines so that the completion time of the job is as low as possible. An algorithmic tool that can be utilized by the designer of such a system is a {\em coordination mechanism} \cite{CKN09}. The coordination mechanism uses a {\em scheduling policy} within each machine that aims to mitigate the impact of selfishness to performance.

We focus on {\em unrelated machine scheduling}. There are $m$ available machines and $n$ players, each controlling a distinct job. The job (owned by player) $u$ has a (possibly infinite) positive processing time (or load) $w_{u,j}$ when processed by machine $j$. A scheduling policy defines the way jobs are scheduled within a machine. In its simplest form, such a policy is {\em non-preemptive} and processes jobs uninterruptedly according to some order. {\em Preemptive} scheduling policies (which is our focus here) do not necessarily have this feature (e.g., they may process jobs in parallel) and may even introduce some idle time.

Naturally, a coordination mechanism induces a game with the job owners as players. Each player has all machines as possible {\em strategies}. The term {\em assignment} is used for a snapshot of the game, where each player has selected a strategy, i.e., it has selected a particular machine to process her job. Given an assignment, the cost a player experiences is the completion time of her job on the machine she has selected. This is well-defined by the scheduling policy of the machine and typically depends on the characteristics of all jobs assigned to the machine. 

Assignments in which no player has any incentive to change her strategy are called {\em pure Nash equilibria} (or, simply, {\em equilibria}). When studying a coordination mechanism, we are interested in bounding the inefficiency of equilibria of the game induced by the mechanism. We use the {\em maximum completion time} among all jobs to measure the social cost. A related quantity is the load of a machine which is defined as the total processing time of the jobs assigned to the machine. The {\em makespan} of an assignment is the maximum load over all machines. Clearly, the makespan of an assignment is a lower bound on the maximum completion time. The {\em price of anarchy} (respectively, {\em price of stability}) of the game induced by a coordination mechanism is defined as the worst (respectively, best) ratio of the maximum completion time over all equilibria over the optimal makespan. 

We prefer mechanisms that induce games that always have equilibria. Furthermore, we would like these equilibria to be easy to find. A highly desirable property that ensures that equilibria can be reached by the players themselves (with best-response play) is the existence of a {\em potential function}. A potential function is defined over all possible assignments and has the property that, in any two assignments differing in the strategy of a single player, the difference of the two values of the potential and the difference of the completion time of the deviating player have the same sign.

Coordination mechanisms for scheduling were introduced by Christodoulou et al.~\cite{CKN09}. Immorlica et al.~\cite{ILMS09} were the first to consider coordination mechanisms in the unrelated machine setting and studied several intuitive mechanisms, including \texttt{ShortestFirst} and \texttt{Makespan}. In \texttt{ShortestFirst}, the jobs in each machine are scheduled non-preemptively, in monotone non-decreasing order of their processing time. Since ties are possible, the mechanism has to distinguish between jobs with identical processing times, e.g., using distinct IDs for the jobs. This is necessary for every deterministic non-preemptive coordination mechanism in order to be well-defined. In contrast, in \texttt{Makespan}, each machine processes the jobs assigned to it ``in parallel'' so that they all have the same completion time. So, no ID information is required by \texttt{Makespan}. We use the term {\em anonymous} to refer to coordination mechanisms having this property. These two coordination mechanisms are {\em strongly local} in the sense that the only information that is required to compute the schedule of jobs within a machine is the processing time of the jobs on that machine only. A {\em local} coordination mechanism may use all parameters of the jobs that are assigned to a machine (e.g., the whole load vector of each job).

Azar et al.~\cite{AFJ+15} prove lower bounds of $\Omega(m)$ and $\Omega(\log{m})$ on the price of anarchy for any strongly local and local non-preemptive coordination mechanism, respectively. On the positive side, they presented two local coordination mechanisms with price of anarchy $o(m)$. Their first coordination mechanism (henceforth called \texttt{AFJMS-1}) is non-preemptive and may induce game without equilibria. When the induced game has equilibria, the price of anarchy is at most $O(\log{m})$. Their second coordination mechanism (henceforth called \texttt{AFJMS-2} is preemptive, induces potential games, and has price of anarchy $O(\log^2{m})$. Both mechanisms are not anonymous.

Caragiannis \cite{C09} presents three more coordination mechanisms. The mechanism \Acoord, induces potential games with price of anarchy $O(\log{m})$. The mechanism uses the distinct IDs of the jobs to ensure that the equilibria of the game are essentially assignments that are reached by a greedy-like online algorithm for minimizing the $p$-norm of machine loads. \cite{AAG+95} and \cite{C08} study this online scheduling problem; the results therein imply that the price of stability of mechanism \Acoord\ is $\Omega(\log{m})$ as well. A different coordination mechanism with similar characteristics (called \balance) is presented in \cite{CDT12}. The coordination mechanism \Bcoord\ (defined also in \cite{C09}) has even better price of anarchy $O\left(\frac{\log{m}}{\log\log{m}}\right)$ (matching a lower bound due to Abed and Huang \cite{AbedH12} for all deterministic coordination mechanisms) but the induced games are not potential ones and may not even have equilibria. However, the price of anarchy bound for \Bcoord\ indicates that preemption may be useful in order to beat the $\Omega(\log{m})$ lower bound for non-preemptive mechanisms from \cite{AFJ+15}. Interestingly, this mechanism is anonymous. The third mechanism \Ccoord\ is anonymous as well, induces potential games, and has price of anarchy and price of stability $O(\log^2m)$ and $O(\log m)$, respectively. To the best of our knowledge, this is the only anonymous mechanism that induces potential games and has polylogarithmic price of anarchy.\footnote{Even though their mechanism \balance\ heavily uses job IDs, Cohen et al.~\cite{CDT12} claim that it is anonymous. This is certainly false according to our terminology since anonymity imposes that two jobs with identical load vectors should be indistinguishable.} Table \ref{tab} summarizes the known local coordination mechanisms.

  \begin{table}[h!]
 \centering
 \caption{A comparison between \Dcoord\ and other local coordination mechanisms from the literature.}
 \begin{tabular}{l l l l l l l l}
\hline
Coordination & PoA & PoS & PNE & Pot.  & IDs & Preempt. & Reference \\
mechanism & & & & & &  & \\
\hline
\hline
\texttt{AFJMS-1} & $\Theta(\log m)$ & - & No & No  & Yes & No & \cite{AFJ+15}\\
\hline
\texttt{AFJMS-2} & $O(\log^2 m)$ & - & Yes & Yes & Yes & Yes & \cite{AFJ+15} \\
\hline
\Acoord & $O(\log m)$ & $\Theta(\log{m})$ & Yes & Yes  & Yes & Yes & \cite{C09} \\
\hline
\balance & $O(\log m)$ & $\Theta(\log{m})$ & Yes & Yes  & Yes & Yes & \cite{CDT12} \\
\hline
\Bcoord & $\Theta(\frac{\log m}{\log \log m})$& - & ? & No  & No & Yes & \cite{C09} \\
\hline
\Ccoord & $O(\log^2 m)$ & $O(\log{m})$ & Yes & Yes  & No & Yes & \cite{C09}\\
\hline
\Dcoord & ${O(\log m)}$ & $O(1)$ & Yes & Yes  & No & Yes & this paper\\
\hline
 \end{tabular}
 \label{tab}
 \end{table}

In the discussion above, we have focused on papers that define the social cost as the maximum completion time (among all players). An alternative social cost that has received much attention is the {\em weighted average completion time}; see \cite{ACH14,BIKM14,CCG+15} for some recent related results. Interestingly, the design principles that lead to efficient mechanisms in their case are considerably different.

Our contribution is as follows. We introduce a quite broad class (called ${\cal M}(d)$) of local anonymous coordination mechanisms that induce potential games. The class contains the coordination mechanism \Ccoord\ as well as the novel coordination mechanism \Dcoord, which has additional {\em almost ideal} properties. In particular, we prove that it has logarithmic price of anarchy and only constant price of stability. A (qualitative and quantitative) comparison of \Dcoord\ to other known local coordination mechanisms is depicted in Table \ref{tab}. 

The rest of the paper is structured as follows. We begin with preliminary definitions in Section \ref{sec:prelim}. Section \ref{sec:broad} is devoted to the definition of the class of mechanisms ${\cal M}(d)$ and to the proof that all mechanisms in this class induce potential games. Then, the novel mechanism \Dcoord\ from this class is defined in Section \ref{sec:dcoord}; its feasibility as well as preliminary statements that are useful for the analysis are also presented there. Finally, in Section \ref{sec:anal}, we prove the bounds on the price of anarchy and stability.

\section{Definitions and preliminaries}\label{sec:prelim}
Throughout the paper, we denote the number of machines by $m$. The index $j$ always refers to a machine; the sum $\sum_j$ runs over all available machines. An assignment is a partition $N=(N_1, ..., N_m)$ of the players to the $m$ machines. So, $N_j$ is the set of players assigned to machine $j$ under $N$. We use the notation $L_j(N_j)$ to refer to the load of machine $j$, i.e., $L_j(N_j)=\sum_{u\in N_j}{w_{u,j}}$.

A coordination mechanism uses a scheduling policy per machine. For every set of jobs assigned to machine $j$, the scheduling policy of the machine defines a detailed schedule of the jobs in the machine, i.e., it defines which job is executed in each point in time, whether more than one jobs are executed in parallel, or whether a machine stays idle for particular time intervals. Instead of defining coordination mechanisms at this level of detail, it suffices to focus on the definition of the completion time ${\cal P}(u,N_j)$ for the job of each player $u\in N_j$. This definition should correspond to some feasible detailed scheduling of jobs in the machine. A sufficient condition that guarantees {\em feasibility} is to define completion times that are never smaller than the machine load.

Like the coordination mechanisms in \cite{AFJ+15,C09,CDT12}, our coordination mechanisms are local. The completion time ${\cal P}(u,N_j)$ of the job belonging to player $u$ in machine $j$ depends on the processing times the jobs in $N_j$ have on machine $j$, as well as on the minimum processing time $w_{u}=\min_j{w_{u,j}}$ of job $u$ over all machines.

Our proofs exploit simple facts about Euclidean norms of machine loads. Recall that, for $p\geq 1$, the $p$-norm of the vector of machine loads $L(N)=(L_1(N_1),L_2(N_2), ..., L_m(N_m))$ under an assignment $N$ is $\norm{L(N)}_p=\left(\sum_j{L_j(N_j)^p}\right)^{1/p}$. By convention, we denote the makespan $\max_{j}{L_j(N_j)}$ as $\norm{L(N)}_\infty$. The following property follows easily by the definition of norms; we use it extensively in the following.

\begin{lemma}\label{lem:norm-max}
For any $p\geq 1$ and any assignment $N$, $\norm{L(N)}_\infty\leq \norm{L(N)}_p \leq m^{1/p}\norm{L(N)}_\infty$.
\end{lemma}
We also use the well-known Minkowski inequality (or triange inequality for the $p$-norm). For machine loads, it reads as follows:
\begin{lemma}[Minkowski inequality]\label{lem:minkowski} 
For every $p \geq 1$ and two assignments $N$ and $N'$, $\norm{L(N)+L(N')}_p\leq \norm{L(N)}_p+\norm{L(N')}_p$. 
\end{lemma}
The notation $L(N)+L(N')$ denotes the $m$-entry vector with $L_j(N_j)+L_j(N'_j)$ at the $j$-th entry. Another necessary technical lemma follows by the convexity properties of polynomials; see \cite{C09} for a proof. 
\begin{lemma}\label{lem:conv}
For $r, t \geq 0$, positive integer $p$, and $a_i\geq 0$ for $i=1, ..., p$, it holds 
\begin{eqnarray*}
\sum_{i=1}^k{\left((t+a_i)^r-t^r\right)} &\leq & \left(t+\sum_{k=1}^p{a_i}\right)^r-t^r.
\end{eqnarray*}
\end{lemma}

\section{A broad class of coordination mechanisms}\label{sec:broad}
In this section, we show that the coordination mechanism \Ccoord\ from \cite{C09} can be thought of as belonging to a broad class of coordination mechanisms, which we call ${\cal M}(d)$. This class contains also our novel coordination mechanism \Dcoord, which will be presented in Section \ref{sec:dcoord}.

The definition of \Ccoord\ uses a positive integer $d\geq 2$ and the functions $\Psi_j$ that map sets of players to the reals as follows. For any machine $j$, $\Psi_j(\emptyset)=0$ and for any non-empty set of players $U=\{u_1, u_2, ..., u_\ell\}$, 
$$\Psi_j(U) = d! \sum_{t_1 +t_2+...+t_\ell=d}{\prod_{k=1}^\ell{w_{u_k,j}^{t_k}}}.$$
The sum runs over all multi-sets of non-negative integers $\{t_1, t_2, ..., t_\ell\}$ that satisfy $t_1+t_2+...+t_\ell=d$. So, $\Psi_j(U)$ is the sum of all possible degree-$d$ monomials of the processing times of the jobs belonging to players from $U$ on machine $j$, with each term in the sum having a coefficient of $d!$. \Ccoord\ schedules the job of player $u_i$ on machine $j$ in an assignment $N$ so that its completion time is 
\begin{eqnarray*}
{\cal P}(u_i,N_j) &=& \left(\frac{w_{u_i,j}\Psi_j(N_j)}{w_{u_i}}\right)^{1/d}.
\end{eqnarray*}
 
We will extend \Ccoord\ to define a broad class of coordination mechanisms; we use ${\cal M}(d)$ to refer to this class, where $d\geq 2$ is a positive integer. Each member of ${\cal M}(d)$ is identified by a {\em coefficient function} $\gamma$. The coefficient functions are defined over multi-sets of non-negative integers that have sum equal to $d+1$ and take non-negative values. An important property of the coefficient functions is that they are {\em invariant to zeros} that requires that for a multi-set $A$ of integers that sum up to $d+1$, $\gamma(A)=\gamma(A\cup \{0\})$. Hence, the value returned by $\gamma$ depends only on the non-zero elements in the multiset it takes as argument.

The definition of a coordination mechanism in ${\cal M}(d)$ uses the quantity $\Lambda_{u_i,j}(U)$, which is defined as follows for a machine $j$ and a job $u_i$ from a subset of jobs $U=\{u_1, u_2, ..., u_\ell\}$:
\begin{eqnarray}\label{eq:Lambda-player}
\Lambda_{u_i,j}(U) &=& \sum_{\begin{subarray}{c} t_1+t_2+...+t_\ell=d+1\\t_i\geq 1\end{subarray}}{\gamma(\{t_1, t_2, ..., t_\ell\})\prod_{k=1}^\ell{w_{u_k,j}^{t_k}}}.
\end{eqnarray}
The sum runs over all multi-sets of non-negative integers, with each integer corresponding to a distinct player of $U$, so that the integer $t_i$ corresponding to player $u_i$ is strictly positive. Notice that $\gamma$ is defined over (unordered) multi-sets; this implies that symmetric monomials have the same coefficient. For example, for the set of players $U=\{u_1,u_2\}$ and a machine $j$, 
$$\Lambda_{u_1,j}(U)=\gamma(\{3,0\})w^3_{{u_1},j}+\gamma(\{2,1\})w^2_{u_1,j}w_{u_2,j}+\gamma(\{1,2\})w_{u_1,j}w^2_{u_2,j}.$$ Clearly, $\{2,1\}$ and $\{1,2\}$ denote the same multi-set and, hence, the coefficients of the (symmetric) second and third monomial are identical.

A coordination mechanism of ${\cal M}(d)$ sets the completion time of player $u_i$ to 
\begin{eqnarray}\label{eq:completion-time}
{\cal P}(u_i,N_j) &=& \left(\frac{\Lambda_{u_i,j}(N_j)}{w_{u_i}}\right)^{1/d}.
\end{eqnarray}
when its job is scheduled on machine $j$ under assignment $N$. 

By simply setting $\gamma(A)=d!$ for every multi-set $A$ of non-negative integers summing up to $d+1$, we obtain \Ccoord. Indeed, it is easy to see that $\Lambda_{u_i,j}(U)=w_{u_i,j}\Psi_j(U)$ in this case.

The definition of ${\cal M}(d)$ guarantees that all its members satisfy two important properties. First, every coordination mechanism in ${\cal M}(d)$ is anonymous. This is due to the fact that the definition of the completion time in (\ref{eq:completion-time}) does not depend on the identity of a player and the jobs of two different players $u$ and $u'$ that have equal processing times $w_{u,j}=w_{u',j}$ at machine $j$ and the same minimum processing time (over all machines) will enjoy identical completion times therein, when each is scheduled together with a set $U$ of other players (i.e., ${\cal P}(u,U\cup \{u\})={\cal P}(u',U\cup \{u'\})$) or when the set of players $N_j$ assigned to machine $j$ contains both $u$ and $u'$ (${\cal P}(u,N_j)={\cal P}(u',N_j)$ in this case).

Another important property of the coordination mechanisms in ${\cal M}(d)$ is that they always induce potential games. We will prove this is a while, after defining the function $\Lambda_j(U)$, again for a machine $j$ and a set of players $U=\{u_1, u_2, ..., u_\ell\}$, as follows:
\begin{eqnarray}\label{eq:Lambda}
\Lambda_{j}(U) &=& \sum_{\begin{subarray}{c} t_1+t_2+...+t_\ell=d+1 \end{subarray}} \gamma(\{t_1,t_2, ..., t_\ell\}) \prod_{k=1}^{\ell} {w}^{t_{k}}_{u_k,j}.
\end{eqnarray}
Compared to the definition of $\Lambda_{u_i,j}(U)$ in (\ref{eq:Lambda-player}), the sum in (\ref{eq:Lambda}) runs just over all multi-sets of non-negative integers (corresponding to players in $U$) that sum up to $d+1$, without any additional constraint.

We will sometimes use the informal term $\Lambda$-functions to refer to the functions defined in both (\ref{eq:Lambda-player}) and (\ref{eq:Lambda}). We can now state and prove the following property of $\Lambda$-functions that we will use several times in our analysis below. For example, it will be particularly useful in order to prove that mechanisms of ${\cal M}(d)$ induce potential games (in Theorem \ref{thm:potential-games}).

\begin{lemma}\label{lem:pot-equality}
Consider a machine $j$ and a set of players $U=\{u_1, u_2, ..., u_\ell\}$. Then, for every player $u_i\in U$,
\begin{eqnarray*}
\Lambda_j(U) &=& \Lambda_{u_i,j}(U) +\Lambda_j(U\setminus \{u_i\}).
\end{eqnarray*}
\end{lemma}

\begin{proof}
Without loss of generality, let us assume $i=1$. Using the definition of $\Lambda$-functions in (\ref{eq:Lambda-player}) and (\ref{eq:Lambda}), we obtain
\begin{eqnarray*}
\Lambda_{j}(U) & = & \sum_{\begin{subarray}{c} t_1+t_2+...+t_\ell=d+1 \end{subarray}}{\gamma(\{t_1, t_2, ..., t_\ell\}) \prod_{k=1}^{\ell}{w^{t_k}_{u_k,j}}}\\
 & = &  \sum_{\begin{subarray}{c} t_1+t_2+...+t_\ell=d+1\\t_1\geq 1 \end{subarray}}{\gamma(\{t_1, t_2, ..., t_\ell\}) \prod_{k=1}^{\ell}{w^{t_k}_{u_k,j}}} \\
& & +  \sum_{\begin{subarray}{c} t_1+t_2+...+t_\ell=d+1\\t_1=0 \end{subarray}}{\gamma(\{t_1, t_2, ..., t_\ell\}) \prod_{k=1}^{\ell}{w^{t_k}_{u_k,j}}}\\
& = &   \Lambda_{u_1,j}(U)   +  \sum_{\begin{subarray}{c} t_2+...+t_\ell=d+1 \end{subarray}}{\gamma(\{t_2, ..., t_\ell\}) \prod_{k=2}^{\ell}{w^{t_k}_{u_k,j}}}\\
& = & \Lambda_{u_{1},j}(U) + \Lambda_{j}(U\setminus \{u_1\}).
\end{eqnarray*}
In the third equality, we have used the fact that the coefficient function is invariant to zeros.
\end{proof}

\begin{theorem}\label{thm:potential-games}
The non-negative function $\Phi$, which is defined over assignments of players to machines as $\Phi(N)=\sum_j{\Lambda_j(N_j)}$, is a potential function for the game induced by any coordination mechanism in ${\cal M}(d)$.
\end{theorem}

\begin{proof}
Consider two assignments $N$ and $N'$ that differ in the assignment of a single player $u$. Assume that player $u$ is assigned to machine $j_1$ and $j_2$ in the assignments $N$ and $N'$, respectively.
Using the definition of function $\Phi$ and Lemma \ref{lem:pot-equality}, we have
\begin{eqnarray*}
\Phi(N)-\Phi(N') &=& \sum_j{\Lambda_j(N_j)} - \sum_j{\Lambda_j(N_j)}\\
&=& \Lambda_{j_1}(N_{j_1})+\Lambda_{j_2}(N_{j_2}) - \Lambda_{j_1}(N'_{j_1})-\Lambda_{j_2}(N'_{j_2})\\
&=& \Lambda_{u,j_1}(N_{j_1})+\Lambda_{j_1}(N_{j_1}\setminus\{u\})+\Lambda_{j_2}(N_{j_2})\\
& & - \Lambda_{j_1}(N'_{j_1})-\Lambda_{u,j_2}(N'_{j_2})-\Lambda_{j_2}(N'_{j_2}\setminus\{u\}).
\end{eqnarray*}
Now observe that $N_{j_1}\setminus\{u\}=N'_{j_1}$ and $N'_{j_2}\setminus\{u\}=N_{j_2}$. Hence, using this observation and the definition of the completion time for $u$ in assignments $N$ and $N'$, the above derivation becomes
\begin{eqnarray*}
\Phi(N)-\Phi(N') &=& \Lambda_{u,j_1}(N_{j_1}) - \Lambda_{u,j_2}(N'_{j_2})\\
&=& w_u\left({\cal P}(u,N_{j_1})^d-{\cal P}(u,N'_{j_2})^d\right),
\end{eqnarray*}
which implies that the difference in the potentials and the difference ${\cal P}(u,N_{j_1})-{\cal P}(u,N'_{j_2})$ in the completion time of the deviating player $u$ in the two assignments have the same sign as desired.
\end{proof}

\section{The coordination mechanism \Dcoord}\label{sec:dcoord}
Like \Ccoord, our new coordination mechanism \Dcoord\ belongs to class ${\cal M}(d)$. It uses the coefficient function defined as 
\begin{eqnarray*}
\gamma(\{t_1, t_2, ..., t_\ell\}) &=& \left\{\begin{array}{ll} 1 & \mbox{if $\exists i$ such that $t_i=d+1$}\\ \frac{d!d}{t_1!t_2!...t_\ell!} & \mbox{otherwise}\end{array}\right.
\end{eqnarray*}
for every multi-set of integers $\{t_1, t_2, ..., t_\ell\}$ such that $t_1+t_2+...+t_\ell=d+1$.

Observe that $\gamma(\{t_1,t_2, ..., t_\ell\})$ is very similar (but not identical) to the multinomial coefficient defined as ${d+1 \choose {t_1, t_2, ..., t_\ell}}=\frac{(d+1)!}{t_1!\ldots t_\ell!}$. This is exploited in the proof of the next statement.

\begin{lemma}\label{lem:lambda-vs-load}
Consider a machine $j$ and a subset of players $U=\{u_1, u_2, ..., u_\ell\}$. Then,
\begin{eqnarray*}
\Lambda_j(U) &=& \frac{d}{d+1}L_j(U)^{d+1}+\frac{1}{d+1}\sum_{u\in U }{w_{u,j}^{d+1}}.
\end{eqnarray*}
\end{lemma}

\begin{proof}
By the definition of $\Lambda_j(U)$ and the coefficient function $\gamma$, we have 
\begin{eqnarray*}
\Lambda_j(U) &=& \sum_{t_1+t_2+...+t_\ell=d+1}{\gamma(\{t_1,t_2, ..., t_\ell\})\prod_{k=1}^\ell{w_{u_k,j}^{t_k}}}\\
&=& \frac{d}{d+1}\sum_{t_1+t_2+...+t_\ell=d+1}{{d+1\choose {t_1, t_2, ..., t_\ell}}\prod_{k=1}^\ell{w_{u_k,j}^{t_k}}}+\frac{1}{d+1}\sum_{u\in U}{w_{u,j}^{d+1}}\\
&=& \frac{d}{d+1}L_j(U)^{d+1}+\frac{1}{d+1}\sum_{u\in U}{w_{u,j}^{d+1}}
\end{eqnarray*}
as desired.
\end{proof}

We proceed with two properties which relate $\Lambda$-functions to machine loads. The first one follows as a trivial corollary of Lemma \ref{lem:lambda-vs-load} after observing that $\sum_{u\in U }{w_{u,j}^{d+1}}\leq L_j(U)^{d+1}$. 

\begin{corollary}\label{cor:lambda-vs-load}
Consider a machine $j$ and a set of players $U$. Then,
\begin{eqnarray*}
\frac{d}{d+1} L_j(U)^{d+1} &\leq & \Lambda_j(U) \leq L_j(U)^{d+1}.
\end{eqnarray*}
\end{corollary}
The second one will be very useful in proving that \Dcoord\ is feasible and in bounding its price of anarchy.

\begin{lemma}\label{lem:sandwitch}
Let $U=\{u_1, ..., u_\ell\}$ be a set of players. For every player $u_i\in U$ and every machine $j$, it holds that
\begin{eqnarray*}
w_{u_i,j}L_j(U)^d &\leq& \Lambda_{u_i,j}(U) \leq  d\cdot w_{u_i,j}L_j(U)^d.
\end{eqnarray*}
\end{lemma}

\begin{proof}
We will first expand the quantities $w_{u_i,j}L_j(U)^d$ and $\Lambda_{u_i,j}(U)$. We have
\begin{eqnarray}\nonumber
\Lambda_{u_1,j}(N_j) &=& \sum_{\begin{subarray}{c} t_1+t_2+...+t_\ell=d+1\\t_1\geq 1\end{subarray}}{\gamma(\{t_1, t_2, ..., t_\ell\})\prod_{k=1}^\ell{w_{u_k,j}^{t_k}}}\\\label{eq:expand1}
&=& w_{u_1,j} \cdot \sum_{\begin{subarray}{c} t_1+t_2+...+t_\ell=d\end{subarray}}{\gamma(\{t_1+1, t_2, ..., t_\ell\})\prod_{k=1}^\ell{w_{u_k,j}^{t_k}}}
\end{eqnarray}
and
\begin{eqnarray}\nonumber
w_{u_1,j} \cdot L_j(N_j) &=& w_{u_1,j} \cdot \left(\sum_{k=1}^\ell{w_{u_k,j}}\right)^d\\\label{eq:expand2}
&=& w_{u_1,j} \cdot \sum_{\begin{subarray}{c} t_1+t_2+...+t_\ell=d\end{subarray}}{{d \choose {t_1, t_2, ..., t_\ell}}\prod_{k=1}^\ell{w_{u_k,j}^{t_k}}}.
\end{eqnarray}

We can prove the two desired inequalities by comparing the corresponding coefficients of each monomial in 
equations (\ref{eq:expand1}) and (\ref{eq:expand2}). Recall that, when $t_1+t_2+...+t_\ell=d$, the coefficient $\gamma(\{t_1+1, t_2, ..., t_\ell\})$ from (\ref{eq:expand1}) is equal to $1$ when $t_1=d$. In this case, the corresponding coefficient in (\ref{eq:expand2}) is ${d \choose {d,0,...,0}}=1$ as well.
Otherwise, 
\begin{eqnarray*}
\gamma(\{t_1+1, t_2, ..., t_\ell\}) &=& \frac{d}{t_1+1}{d\choose {t_1,t_2, ..., t_\ell}}.
\end{eqnarray*}
Since $t_1$ is non-negative and at most $d-1$, we have that 
\begin{eqnarray*}
{d\choose {t_1,t_2, ..., t_\ell}} &\leq & \gamma(\{t_1+1, t_2, ..., t_\ell\}) \leq d \cdot {d\choose {t_1,t_2, ..., t_\ell}},
\end{eqnarray*}
which concludes the proof.
\end{proof}

Feasibility follows easily now.

\begin{theorem}\label{thm:feasible}
\Dcoord\ produces feasible schedules. 
\end{theorem}

\begin{proof}
Consider player $u_1$ and any assignment $N$ which assigns it to machine $j$ together with $\ell-1$ other players $u_2, u_3, ..., u_\ell$. By the leftmost inequality of Lemma \ref{lem:sandwitch}, we have that
\begin{eqnarray*}
{\cal P}(u_1, N_j) &=& \left(\frac{\Lambda_{u_1,j}(N_j)}{w_{u_1}}\right)^{1/d} \geq \left(\frac{w_{u_1,j}}{w_{u_1}}\right)^{1/d}L_j(N_j) \geq  L_j(N_j),
\end{eqnarray*}
as desired. The inequality holds since, by definition, $w_{u_1,j}\geq w_{u_1}$.
\end{proof}

\section{Bounding the price of anarchy and stability}\label{sec:anal}
For proving the price of anarchy bound, we will need the following lemma which relates the load of any machine at an equilibrium with the optimal makespan.

\begin{lemma}\label{lem:complicated}
Let $N$ be an equilibrium and $N^*$ an assignment of optimal makespan. Then, for every machine $j$, it holds that
$$L_j(N) \leq m^{\frac{1}{d+1}}\frac{d+1}{\ln{2}}\norm{L(N^*)}_{\infty}.$$
\end{lemma}

\begin{proof}
Consider a player $u$ that is assigned to machine $j$ in the equilibrium assignment $N$ and to machine $j'$ in the assignment $\opt{N}$ that minimizes the $l_{d+1}$-norm of the machine loads. First, consider the case where $j\not=j'$. In the equilibrium assignment $N$, player $u$ has no incentive to deviate from machine $j$ to machine $j'$ and, hence, ${\cal P}(u,N_j)\leq {\cal P}(u,N_{j'}\cup \{u\})$. By the definition of \Dcoord, we obtain that $\Lambda_{u,j}(N_j) \leq \Lambda_{u,j'}(N_{j'}\cup \{u\})$. Using this observation, Lemma \ref{lem:pot-equality}, and Lemma \ref{lem:lambda-vs-load}, we get
\begin{eqnarray*}
\Lambda_{u,j}(N_j) &\leq & \Lambda_{u,j'}(N_{j'}\cup \{u\}) = \Lambda_{j'}(N_{j'}\cup \{u\}) - \Lambda_{j'}(N_{j'})\\
&=& \frac{d}{d+1}L_{j'}(N_{j'}\cup \{u\})^{d+1} - \frac{d}{d+1}L_{j'}(N_{j'})^{d+1}+\frac{1}{d+1}w_{u,j'}^{d+1}\\
&=& \frac{d}{d+1}(L_{j'}(N_{j'})+w_{u,j'})^{d+1}-\frac{d}{d+1}L_{j'}(N_{j'})^{d+1}+\frac{1}{d+1}w_{u,j'}^{d+1}
\end{eqnarray*}
We will now prove that the same inequality holds when $j=j'$. In this case, together with Lemmas \ref{lem:pot-equality} and \ref{lem:lambda-vs-load}, we need to use a different argument that exploits a convexity property of polynomials. We have 
\begin{eqnarray*}
\Lambda_{u,j}(N_j) &= & \Lambda_{u,j'}(N_{j'}) = \Lambda_{j'}(N_{j'}) - \Lambda_{j'}(N_{j'}\setminus\{u\})\\
&=& \frac{d}{d+1}L_{j'}(N_{j'})^{d+1} - \frac{d}{d+1}L_{j'}(N_{j'}\setminus\{u\})^{d+1}+\frac{1}{d+1}w_{u,j'}^{d+1}\\
&=& \frac{d}{d+1}L_{j'}(N_{j'})^{d+1} - \frac{d}{d+1}(L_{j'}(N_{j'})-w_{u,j'})^{d+1}+\frac{1}{d+1}w_{u,j'}^{d+1}\\
&\leq & \frac{d}{d+1}(L_{j'}(N_{j'})+w_{u,j'})^{d+1}-\frac{d}{d+1}L_{j'}(N_{j'})^{d+1}+\frac{1}{d+1}w_{u,j'}^{d+1}.
\end{eqnarray*}
The last inequality follows since $z^{d+1}-(z-\alpha)^{d+1} \leq (z+\alpha)^{d+1}-z^{d+1}$ for every $z\geq \alpha\geq 0$, due to the convexity of the polynomial function $z^{d+1}$.

Let us sum the above inequality over all players. We obtain
\begin{eqnarray}\nonumber
& & \sum_j{\sum_{u\in N_j}{\Lambda_{u,j}(N_j)}}\\\nonumber
&\leq & \frac{d}{d+1}\sum_{j}{\sum_{u\in \opt{N}_{j}}{\left(\left(L_{j}(N_{j})+w_{u,j}\right)^{d+1}-L_{j}(N_{j})^{d+1}\right)}+\frac{1}{d+1}\sum_{j}\sum_{u\in \opt{N}_{j}}{{w_{u,j}^{d+1}}}}\\\nonumber
&\leq & \frac{d}{d+1}\sum_{j}{\left(\left(L_{j}(N_{j})+\sum_{u\in \opt{N}_{j}}{w_{u,j}}\right)^{d+1}-L_{j}(N_{j})^{d+1}\right)}+\frac{1}{d+1}\sum_{j}{L_{j}(\opt{N}_{j})^{d+1}}\\\nonumber
&=& \frac{d}{d+1}\norm{L(N)+L(\opt{N})}_{d+1}^{d+1} - \frac{d}{d+1}\norm{L(N)}_{d+1}^{d+1}+\frac{1}{d+1}\norm{L(\opt{N})}_{d+1}^{d+1}\\\label{eq:ineq1}
&\leq & \frac{d}{d+1}\left(\norm{L(N}_{d+1}+\norm{L(\opt{N})}_{d+1}\right)^{d+1}-\frac{d-1}{d+1}\norm{L(N)}_{d+1}^{d+1}.
\end{eqnarray}
The second inequality follows by Lemma \ref{lem:conv} and since $\sum_{u\in \opt{N}_{j}}{w_{u,j}^{d+1}}\leq L_{j}(\opt{N})^{d+1}$. The equality follows by the definition of $l_{d+1}$-norms and the last inequality follows by Minkowski inequality (Lemma \ref{lem:minkowski}) and by the fact that $\norm{L(\opt{N})}\leq \norm{L(N)}$.

Using the definition of norms and Lemma \ref{lem:lambda-vs-load}, we also have
\begin{eqnarray}\label{eq:ineq2}
\norm{L(N)}_{d+1}^{d+1} &=& \sum_j{L_j(N_j)^{d+1}} = \sum_j{\sum_{u\in N_j}{w_{u,j}L_j(N_j)^d}} \leq  \sum_j{\sum_{u\in N_j}{\Lambda_{u,j}}}.
\end{eqnarray}

By combining (\ref{eq:ineq1}) and (\ref{eq:ineq2}), we have 
\begin{eqnarray*}
2\norm{L(N)}_{d+1}^{d+1} &\leq & \left(\norm{L(N}_{d+1}+\norm{L(\opt{N})}_{d+1}\right)^{d+1}
\end{eqnarray*}
and, equivalently, 
\begin{eqnarray*}
\norm{L(N)}_{d+1} &\leq & \frac{1}{2^{\frac{1}{d+1}}-1}\norm{L(\opt{N})}_{d+1} \leq \frac{d+1}{\ln{2}}\norm{L(N^*)}_{d+1}\\
&\leq & m^{\frac{1}{d+1}} \frac{d+1}{\ln{2}}\norm{L(N^*)}_{\infty}.
\end{eqnarray*}
The second inequality follows since, by definition, $\norm{L(\opt{N})}_{d+1}\leq \norm{L(N^*)}_{d+1}$ and by the inequality $e^z\geq z+1$. The third inequality follows by Lemma \ref{lem:norm-max}. Since $\norm{L(N)}_{d+1}\geq L_j(N_j)$ for every machine $j$, the lemma follows.
\end{proof}

For the price of stability bound, we will use a qualitative similar (to Lemma \ref{lem:complicated}) relation between machine loads at a particular equilibrium and the optimal makespan.

\begin{lemma}\label{lem:potential-method}
Let $N$ be the equilibrium that minimizes the potential function and $N^*$ an assignment of optimal makespan. Then, for every machine $j$, it holds that
$$L_j(N_j) \leq \left(\frac{d+1}{d}m\right)^{\frac{1}{d+1}} \norm{L(N^*)}_\infty.$$
\end{lemma}

\begin{proof}
Observe that $\Phi(N)\leq \Phi(N^*)$ since every equilibrium that is reached when players repeatedly best-respond starting from assignment $N^*$ has potential at most $\Phi(N^*)$. Using this observation, the definition of norms, Corollary \ref{cor:lambda-vs-load}, and the definition of the potential function (see the statement of Theorem \ref{thm:potential-games}), we have
\begin{eqnarray*}
\norm{L(N)}|_{d+1}^{d+1} &=& \sum_j{L_j(N_j)^{d+1}}  \leq \frac{d+1}{d}\sum_j{\Lambda_j(N_j)} = \frac{d+1}{d}\Phi(N)\\
&\leq & \frac{d+1}{d}\Phi(N^*)= \frac{d+1}{d} \sum_j{\Lambda_j(N^*_j)^{d+1}} \leq \frac{d+1}{d}\sum_j{L_j(N^*_j)^{d+1}}\\
&=& \frac{d+1}{d}\norm{L(N^*)}_{d+1}^{d+1}.
\end{eqnarray*}
Hence, for every machine $j$, by exploiting Lemma \ref{lem:norm-max}, we have $L_j(N)\leq \norm{L(N)}_{d+1} \leq \left(\frac{d+1}{d}\right)^{\frac{1}{d+1}}\norm{L(N^*)}_{d+1} \leq \left(\frac{d+1}{d}m\right)^{\frac{1}{d+1}}\norm{L(N^*)}_\infty$ as desired. 
\end{proof}

We are now ready to complete the price of anarchy/stability proofs. We will do so by comparing the completion time of any player to the optimal makespan $\norm{L(N^*)}_\infty$.

\begin{theorem}
By setting $d=O(\log{m})$, \Dcoord\ has price of anarchy $O(\log{m})$ and price of stability $O(1)$.
\end{theorem}

\begin{proof}
Consider a player $u$ that is assigned to machine $j$ at some equilibrium $N$ and satisfies $w_u=w_{u,j^*}$ for some machine $j^*$. We will use the fact that player $u$ (is either already at or) has not incentive to deviate to machine $j^*$ at equilibrium to bound its completion time as follows:
\begin{eqnarray*}
{\cal P}(u,N_j) &\leq & {\cal P}(u,N_{j^*}\cup \{u\}) = \left(\frac{\Lambda_{u,j^*}(N_{j^*}\cup \{u\})}{w_u}\right)^{1/d}\\
&\leq & \left(\frac{dw_{u,j^*}L_{j^*}(N_{j^*}\cup \{u\})^d}{w_u}\right)^{1/d} \leq d^{1/d} (L_{j^*}(N_{j^*})+w_u).
\end{eqnarray*}
The equality follows by the definition of \Dcoord, and the second inequality follows by Lemma \ref{lem:sandwitch}. The third inequality follows since $w_u=w_{u,j^*}$ and by observing that $L_{j^*}(N_{j^*}\cup \{u\})=L(N_{j^*})+w_u$ if $u\not\in N_{j^*}$ (i.e., $j\not=j^*$) and $L_{j^*}(N_{j^*}\cup \{u\})=L(N_{j^*})$ otherwise. 

Now, using Lemma \ref{lem:complicated} to bound $L_{j^*}(N_{j^*})$, we obtain that
\begin{eqnarray*}
{\cal P}(u,N_j) &\leq & d^{1/d}\left(m^{\frac{1}{d+1}}\frac{d+1}{\ln{2}}+1\right)\norm{L(N^*)}_\infty.
\end{eqnarray*} 
If the equilibrium $N$ is a potential-minimizing assignment, Lemma \ref{lem:potential-method} can be further used to obtain the better guarantee
\begin{eqnarray*}
{\cal P}(u,N_j) &\leq & d^{1/d}\left(\left(\frac{d+1}{d}m\right)^{\frac{1}{d+1}}+1\right)\norm{L(N^*)}_\infty.
\end{eqnarray*} 
The theorem follows since, by setting $d=\Theta(\log{m})$, the factors (ignoring $\norm{L(N^*)}_\infty$ in the rightmost expressions become $O(\log{m})$ and $O(1)$, respectively. So, in general, we have that the completion time of any player at equilibrium is at most $O(\log{m})$ times the optimal makespan (hence, the price of anarchy bound) while there exists a particular equilibrium where the completion time of any player is at most $O(1)$ times the optimal makespan (hence, the price of stability bound).
\end{proof}




\end{document}